\documentclass{article}

     \usepackage[preprint,nonatbib]{neurips_2021}

\usepackage[utf8]{inputenc} 
\usepackage[T1]{fontenc}    
\usepackage{hyperref}       
\usepackage{url}            
\usepackage{booktabs}       
\usepackage{amsfonts}       
\usepackage{nicefrac}       
\usepackage{microtype}      
\usepackage{xcolor}         
\usepackage{amsmath}
\usepackage{subfigure}
\usepackage{graphicx}
\usepackage{multirow}

\title{Sublinear Time Nearest Neighbor Search over Generalized Weighted Manhattan Distance}

\author{%
  Huan Hu \\
  Harbin Institute of Technology, China\\
  \texttt{hit\_huhuan@foxmail.com}\\
  \And
  Jianzhong Li \\
  Harbin Institute of Technology, China \\
  \texttt{lijzh@hit.edu.cn} \\
}

\begin{document}

\maketitle

\begin{abstract}
  Nearest Neighbor Search (NNS) over generalized weighted distances is fundamental to a wide range of applications. 
  The problem of NNS over the generalized weighted square Euclidean distance has been studied in previous work.
  However, numerous studies have shown that the Manhattan distance could be more effective than the Euclidean distance for high-dimensional NNS, which indicates that the generalized weighted Manhattan distance is possibly more practical than the generalized weighted square Euclidean distance in high dimensions.
  To the best of our knowledge, 
  no prior work solves the problem of NNS over the generalized weighted Manhattan distance in sublinear time. 
  This paper achieves the goal by proposing two novel hashing schemes ($d_w^{l_1},l_2$)-ALSH and ($d_w^{l_1},\theta$)-ALSH. 
\end{abstract}
\section{Introduction}
\label{introduction}
Nearest Neighbor Search (NNS) over generalized weighted distances is fundamental to a wide variety of applications, such as personalized recommendation \cite{DBLP:conf/www/GuZHS16,DBLP:conf/iccci/HwangST10,DBLP:conf/sigir/McAuleyTSH15} and kNN classification \cite{DBLP:journals/pr/BhattacharyaGC17,DBLP:journals/tkde/MoreoES20}.
Given a set of $n$ data points $D\subset\mathbb{R}^d$ and a query point $q\in\mathbb{R}^d$ with a weight vector $w\in\mathbb{R}^d$, NNS over a generalized weighted distance, denoted by $d_w$, is to find a point $o^*\in D$ such that $o^*$ is the closest data point to $q$ for $d_w$. Formally, the goal of NNS over $d_w$ is to return
\begin{equation}
o^*=\arg\min_{o\in D}d_w(o,q).
\end{equation}
Note that the weight vector $w$ is specified along with $q$ rather than pre-specified. Moreover, each element of $w$ can be either positive or non-positive.

The generalized weighted Manhattan distance, denoted by $d_w^{l_1}$, and the generalized weighted square Euclidean distance, denoted by $d_w^{l_2}$, are two typical generalized weighted distances which are derived from the Manhattan distance and the Euclidean distance, respectively. For any two points $o=(o_1,o_2,\ldots,o_d)\in\mathbb{R}^d$ and $q=(q_1,q_2,\ldots,q_d)\in\mathbb{R}^d$, the distances $d_w^{l_1}(o,q)$ and $d_w^{l_2}(o,q)$ are respectively computed as follows:
\begin{equation}
\begin{split}
d_w^{l_1}(o,q)=&\sum_{i=1}^{d}w_i\left|o_i-q_i\right|\\
d_w^{l_2}(o,q)=&\sum_{i=1}^{d}w_i\left(o_i-q_i\right)^2,
\end{split}
\end{equation}
where $w=(w_1,w_2,\ldots,w_d)\in\mathbb{R}^d$.
A recent paper \cite{DBLP:conf/icml/LeiHKT19} studied the problem of NNS over $d_w^{l_2}$ and provided two sublinear time solutions for it. However, to the best of our knowledge, there is no prior work that solves the problem of NNS over $d_w^{l_1}$ in sublinear time.
Actually, plenty of studies \cite{DBLP:conf/icdt/AggarwalHK01,DBLP:conf/vldb/HinneburgAK00} have shown that the Manhattan distance could be more effective than the Euclidean distance for producing meaningful NNS results in high-dimensional spaces. It indicates
that NNS over $d_w^{l_1}$ is possibly more practical than NNS over $d_w^{l_2}$ in many real scenarios. In this paper, we target to propose sublinear time methods for efficiently solving the problem of NNS over $d_w^{l_1}$.

As a matter of fact, existing methods can not handle NNS over $d_w^{l_1}$ well. 
Specifically, the brute-force linear scan scales linearly with data size and thus may yield unsatisfactory performance.
The conventional spatial index-based methods \cite{DBLP:journals/cacm/Bentley75,DBLP:journals/sigmod/CheungF98,samet2006foundations} can only perform well for NNS in low dimensions due to the ``curse of dimensionality'' \cite{DBLP:reference/db/Chen09}. 
Locality-Sensitive Hashing (LSH) \cite{DBLP:journals/corr/WangSSJ14} is a popular approach for approximate NNS and exhibits good performance for high-dimensional cases. In the literature, a number of efficient LSH schemes \cite{DBLP:conf/compgeom/DatarIIM04,DBLP:conf/sigmod/GanFFN12,DBLP:conf/vldb/GionisIM99,DBLP:journals/vldb/HuangFFNW17,DBLP:conf/stoc/IndykM98,DBLP:journals/pvldb/LuWWK20,DBLP:journals/tods/TaoYSK10,DBLP:journals/pvldb/ZhengZWHLJ20} have been proposed based on LSH families, and some of them can answer NNS queries even in sublinear time. Unfortunately, they can not be applied to answer the NNS queries over $d_w^{l_1}$ unless $w$ is fixed to an all-1 vector.

Recently, Asymmetric Locality-Sensitive Hashing (ALSH) was extended from LSH so that the problems of Maximum Inner Product Search (MIPS) and NNS over $d_w^{l_2}$ can be addressed in sublinear time \cite{DBLP:conf/icml/LeiHKT19,DBLP:conf/alt/NeyshaburMS14,DBLP:conf/icml/NeyshaburS15,DBLP:conf/nips/Shrivastava014,DBLP:conf/uai/Shrivastava015}. An ALSH scheme relies on an ALSH family. As far as we know, there is no ALSH family proposed for NNS over $d_w^{l_1}$ in previous work. 
To provide sublinear time solutions for NNS over $d_w^{l_1}$ in this paper, we follow the ALSH approach to propose ALSH schemes by introducing ALSH families that are suitable for NNS over $d_w^{l_1}$.


\textbf{Outline.} In Section \ref{preliminaries}, we review the approaches of LSH and ALSH. In Section \ref{negative results}, we show that there is no LSH or ALSH family for NNS over $d_w^{l_1}$ over the entire space $\mathbb{R}^d$ and there is no LSH family for NNS over $d_w^{l_1}$ over the bounded spaces in $\mathbb{R}^d$. 
Then we seek to find ALSH families for NNS over $d_w^{l_1}$ over the bounded spaces in $\mathbb{R}^d$. As a result, we propose two suitable ALSH families and further obtain two sublinear time ALSH schemes ($d_w^{l_1},l_2$)-ALSH and ($d_w^{l_1},\theta$)-ALSH in Section \ref{our solution}. 
\section{Preliminaries}
\label{preliminaries}
Before introducing our proposed solutions to the problem of NNS over $d_w^{l_1}$, we first present the preliminaries on LSH and ALSH.
\subsection{Locality-Sensitive Hashing}
\label{lsh}
Let $d(\cdot,\cdot)$ be a distance function and $\mathcal{Z}$ be the space where $d(\cdot,\cdot)$ is defined.
Assume that data points and query points are located in $\mathcal{X}\subseteq\mathcal{Z}$ and $\mathcal{Y}\subseteq\mathcal{Z}$, respectively.
Then, an LSH family is formally defined as follows.
\newtheorem{definition}{Definition}
\begin{definition}[LSH Family]
	\label{lsh family}
	An LSH family $\mathcal{H}_{(h)}=\{h:\mathcal{Z}\rightarrow BucketIDs\}$ is called $(R_1,R_2,P_1,P_2)$-sensitive if for any $o\in\mathcal{X}$ and $q\in\mathcal{Y}$, the following conditions are satisfied:
	\vspace{-0.3cm}
	\begin{itemize}
		\setlength{\parskip}{2pt} \setlength{\itemsep}{0pt plus 1pt}
		\item If $d(o,q)\leq R_1$, then Pr$[h(o)=h(q)]\geq P_1$;
		\item If $d(o,q)\geq R_2$, then Pr$[h(o)=h(q)]\leq P_2$;
		\item $R_1<R_2$ and $P_1>P_2$.
	\end{itemize}
	\vspace{-0.3cm}
\end{definition}
As we can see from Definition \ref{lsh family}, an LSH family is essentially a set of hash functions that can hash closer points into the same bucket with higher probability. 
Thus, the basic idea of an LSH scheme is to use an LSH family to hash points such that only the data points that have the same hash code as the query point are likely to be retrieved to find approximate nearest neighbors. 
In the following, we review two popular LSH families that were proposed for the $l_2$ distance (a.k.a. the Euclidean distance) and the Angular distance, respectively.

The $l_2$ distance between any two points $o,q\in\mathbb{R}^d$ is computed as $d^{l_2}(o,q)=\Vert o-q\Vert_2$, where $\Vert\cdot\Vert_2$ is the $l_2$-norm of a vector.
The LSH family proposed for the $l_2$ distance in \cite{DBLP:conf/compgeom/DatarIIM04} is
$\mathcal{H}_{(h^{l_2})}=\{h^{l_2}:\mathbb{R}^d\rightarrow \mathbb{Z}\}$, where
\begin{equation}
\label{lp/e2lsh}
h^{l_2}(x)=\lfloor\frac{a^Tx+b}{w}\rfloor,
\end{equation}
$a$ is a $d$-dimensional vector where each element is chosen independently from the standard normal distribution, $b$ is a real number chosen uniformly at random from $[0,w]$, and $w$ is a user-specified positive constant. Let $r=d^{l_2}(o,q)$. The collision probability function is
\begin{equation}
\label{collisonl2}
P^{l_2}(r)=Pr[h^{l_2}(o)=h^{l_2}(q)]=1-2\Phi(-w/r)-\frac{2}{\sqrt{2\pi}(w/r)}(1-e^{-(w^2/2r^2)}),
\end{equation}
where $\Phi(\cdot)$ is the cumulative distribution function of the standard normal distribution \cite{DBLP:conf/compgeom/DatarIIM04}.

The Angular distance between any two points $o,q\in\mathbb{R}^d$ is computed as $d^{\theta}(o,q)=\arccos(\frac{o^T q}{\Vert o\Vert_2\Vert q\Vert_2})$.
The LSH family proposed for the Angular distance in \cite{DBLP:conf/stoc/Charikar02} is
$\mathcal{H}_{(h^{\theta})}=\{h^{\theta}:\mathbb{R}^d\rightarrow \{0,1\}\}$, where
\begin{equation}
\label{theta/e2lsh}
h^{\theta}(x)=\begin{cases}
0&\text{if}\ a^Tx<0\\
1&\text{if}\ a^Tx\geq0
\end{cases}
\end{equation}
and $a$ is a $d$-dimensional vector where each element is chosen independently from the standard normal distribution. Let $r=d^{\theta}(o,q)$. The collision probability function is
\begin{equation}
\label{collisontheta}
P^{\theta}(r)=Pr[h^{\theta}(o)=h^{\theta}(q)]=1-\frac{r}{\pi}.
\end{equation}
\subsection{Asymmetric Locality-Sensitive Hashing}
Recent studies have shown that ALSH is an effective approach for solving the problems of MIPS and NNS over $d_w^{l_2}$ \cite{DBLP:conf/icml/LeiHKT19,DBLP:conf/icml/NeyshaburS15,DBLP:conf/nips/Shrivastava014,DBLP:conf/uai/Shrivastava015}. 
An ALSH scheme processes NNS queries in a similar way to an LSH scheme.
It relies on an ALSH family.
Formally, the definition of an ALSH family is as follows.
\begin{definition}[ALSH Family]
	\label{alsh family}
	An ALSH family $\mathcal{H}_{(f,g)}=\{f:\mathcal{Z}\rightarrow BucketIDs\}\bigcup\{g:\mathcal{Z}\rightarrow BucketIDs\}$ is called $(R_1,R_2,P_1,P_2)$-sensitive if for any data point $o\in\mathcal{X}$ and query point $q\in\mathcal{Y}$, the following conditions are satisfied:
	\vspace{-0.3cm}
	\begin{itemize}
		\setlength{\parskip}{2pt} \setlength{\itemsep}{0pt plus 1pt}
		\item If $d(o,q)\leq R_1$, then Pr$[f(o)=g(q)]\geq P_1$;
		\item If $d(o,q)\geq R_2$, then Pr$[f(o)=g(q)]\leq P_2$;
		\item $R_1<R_2$ and $P_1>P_2$. 
	\end{itemize}
	\vspace{-0.3cm}
\end{definition}

From Definition \ref{alsh family}, we can see that an ALSH family $\mathcal{H}_{(f,g)}$ consists of a set of hash functions $\{f:\mathcal{Z}\rightarrow BucketIDs\}$ for data points and a set of hash functions $\{g:\mathcal{Z}\rightarrow BucketIDs\}$ for query points, and it ensures that each query point can collide with closer data points with higher probability.
In practice, $\mathcal{H}_{(f,g)}$ is often implemented with an LSH family $\mathcal{H}_{(h')}=\{h':\mathcal{Z}'\rightarrow BucketIDs\}$ and two vector functions called Preprocessing Transformation $P:\mathcal{X}\rightarrow \mathcal{X}'$ and Query Transformation $Q:\mathcal{Y}\rightarrow \mathcal{Y}'$ respectively \cite{DBLP:conf/icml/LeiHKT19,DBLP:conf/icml/NeyshaburS15,DBLP:conf/nips/Shrivastava014,DBLP:conf/uai/Shrivastava015} (here, $\mathcal{X}'\subseteq\mathcal{Z}'$ and $\mathcal{Y}'\subseteq\mathcal{Z}'$). 
Thus, the hash value of each data point $o\in\mathcal{X}$ is computed as $f(o)=h'(P(o))$ and the hash value of each query point $q\in\mathcal{Y}$ is computed as $g(q)=h'(Q(q))$.

Fundamentally, both LSH and ALSH schemes obtain approximate nearest neighbors by efficiently solving the ($R_1,R_2$)-Near Neighbor Search (($R_1,R_2$)-NNS) problem as follows.
\begin{definition}[(\boldmath$R_1,R_2$)-NNS]
	\label{$(R_1,R_2)$-NNS}
	Given a distance function $d(\cdot,\cdot)$, two distance thresholds $R_1$ and $R_2$ ($R_1<R_2$) and a data set $D\subset\mathcal{X}$, for any query point $q\in\mathcal{Y}$, 
	the $(R_1,R_2)$-NNS problem is to return a point $o\in D$ satisfying $d(o,q)\leq R_2$ if there exists a point $o'\in D$ satisfying $d(o',q)\leq R_1$.
\end{definition}
The theorem below indicates that the $(R_1,R_2)$-NNS problem can be solved with an LSH or ALSH scheme in sublinear time.
\newtheorem{theorem}{Theorem}
\begin{theorem}
	\label{sublinear}
	\cite{DBLP:conf/compgeom/DatarIIM04,DBLP:conf/icml/LeiHKT19,DBLP:conf/icml/NeyshaburS15,DBLP:conf/nips/Shrivastava014,DBLP:conf/uai/Shrivastava015} Given an $(R_1,R_2,P_1,P_2)$-sensitive LSH or ALSH family, one can construct a data structure for solving the problem of $(R_1,R_2)$-NNS with $O(n^{\rho}d\log n)$ query time and $O(n^{1+\rho})$ space, where $\rho=\frac{\log P_1}{\log P_2}<1$.
\end{theorem}
\section{Negative Results}
\label{negative results}
In this section, we present some negative results on the existence of LSH and ALSH families for NNS over $d_w^{l_1}$.

The following theorem indicates that it is impossible to find an LSH or ALSH family for NNS over $d_w^{l_1}$ over $\mathcal{X}=\mathcal{Y}=\mathbb{R}^d$ ($d\geq3$).
\newenvironment{proof}{{\noindent\it Proof.}\ }{\hfill $\square$\par}
\begin{theorem}
	\label{no lsh/alsh over Rd}
	For any $d\geq3$, $R_1<R_2$ and $P_1>P_2$, there is no $(R_1,R_2,P_1,P_2)$-sensitive LSH or ALSH family for NNS over $d_w^{l_1}$ over $\mathcal{X}=\mathcal{Y}=\mathbb{R}^d$.
\end{theorem}
\begin{proof}
	An LSH (or ALSH) family for NNS over $d_w^{l_1}$ over $\mathbb{R}^d$ ($d>3$) is also an LSH (or ALSH) family for NNS over $d_w^{l_1}$ over a three-dimensional subspace, i.e., over $\mathbb{R}^3$.
	Hence, we only need to prove that there is no LSH or ALSH family for NNS over $d_w^{l_1}$ over $\mathbb{R}^3$.
	Assume by contradiction that for some $R_1<R_2$ and $P_1>P_2$ there exists an $(R_1,R_2,P_1,P_2)$-sensitive LSH family $\mathcal{H}_{(h)}$ or ALSH family $\mathcal{H}_{(f,g)}$ for NNS over $d_w^{l_1}$ over $\mathbb{R}^3$.
	Consider a set of $N$ data points $\{o^1,o^2,\ldots,o^N\}\subset\mathbb{R}^3$ and a set of $N$ query points $\{q^1,q^2,\ldots,q^N\}\subset\mathbb{R}^3$, where
	for $1\leq i,j\leq N$,
	\begin{equation}
	\begin{aligned}
	o^i&=(iR_1-iR_2, 0, 0)\\
	q^j&=(0, jR_1-jR_2, R_1).
	\end{aligned}
	\end{equation}
	The weight vector specified along with each query point is set as follows:
	\begin{equation}
	w=\begin{cases}
	(1,-1,-1)&\text{if}\ R_1<0\\
	(1,-1,1)&\text{if}\ R_1\geq0.
	\end{cases}
	\end{equation}
	Thus, $d_w^{l_1}(o^i,q^j)=(i-j)(R_2-R_1)+R_1$ for $1\leq i,j\leq N$. As can be seen, $d_w^{l_1}(o^i,q^j)\leq R_1$ if $1\leq i\leq j\leq N$ and $d_w^{l_1}(o^i,q^j)\geq R_2$ if $1\leq j<i\leq N$. Let $A\in\mathbb{R}^{N\times N}$ be a sign matrix where each element is
	\begin{equation}
	A(i,j)=\begin{cases}
	1&\text{if}\ d_w^{l_1}(o^i,q^j)\leq R_1\\
	-1&\text{if}\ d_w^{l_1}(o^i,q^j)\geq R_2\\
	0&\text{otherwise}.
	\end{cases}
	\end{equation}
	Obviously, $A$ is triangular with +1 on and above the diagonal and -1 below it.
	Consider also the matrix $B\in\mathbb{R}^{N\times N}$ of collision probabilities $B(i,j)=Pr[h(o^i)=h(q^j)]$ (for $\mathcal{H}_{(h)}$) or $B(i,j)=Pr[f(o^i)=g(q^j)]$ (for $\mathcal{H}_{(f,g)}$). Let $\theta=\frac{P_1+P_1}{2}<1$ and $\epsilon=\frac{P_1-P_2}{2}>0$. It is easy to know that $A(i,j)(B(i,j)-\theta)\geq\epsilon$ for $1\leq i,j\leq N$. That is,
	\begin{equation}
	\label{mn}
	A\odot\frac{B-\theta}{\epsilon}\geq1,
	\end{equation}
	where $\odot$ denotes the Hadamard (element-wise) product.
	From \cite{DBLP:conf/colt/SrebroS05}, the margin complexity of the sign matrix $A$ is $mc(A)=\inf_{A\odot C\geq1}\Vert C\Vert_{max}$, where $\Vert\cdot\Vert_{max}$ is the max-norm of a matrix. Since $A$ is also an $N\times N$ triangular matrix, the margin complexity of $A$ is bounded by $mc(A)=\Omega(\log N)$ according to \cite{DBLP:journals/ml/ForsterSSS03}. 
	Therefore, from Equation \ref{mn}, we can obtain 
	\begin{equation}
	\label{xyz}
	\Vert\frac{B-\theta}{\epsilon}\Vert_{max}=\Omega(\log N).
	\end{equation}
	Since $B$ is a collision probability matrix,
	the max-norm of $B$ satisfies $\Vert B\Vert_{max}\leq 1$ \cite{DBLP:conf/alt/NeyshaburMS14}. Shifting $B$ by $0<\theta<1$ changes $\Vert B\Vert_{max}$ by at most $\theta$. Thus, we have
	\begin{equation}
	\label{abc}
	\Vert B-\theta\Vert_{max}<2.
	\end{equation}
	Combining Equations \ref{xyz} and \ref{abc}, we can easily derive that $\epsilon=O(\frac{1}{\log N})$. 
	For any $\epsilon=\frac{P_1-P_2}{2}>0$, we get a contradiction by selecting a large enough $N$.
\end{proof}

The proof of Theorem \ref{no lsh/alsh over Rd} is similar to that of Theorem 3.1 in \cite{DBLP:conf/icml/NeyshaburS15}. 
Due to space limitations, for the details of the max-norm and margin complexity involved in the proof of Theorem \ref{no lsh/alsh over Rd}, please refer to \url{http://proceedings.mlr.press/v37/neyshabur15-supp.pdf}.

Actually, in real scenarios data points and query points are usually located in bounded spaces.
Consider the typical case of $\mathcal{X}=\mathcal{Y}=\mathcal{S}$, where $\mathcal{S}\subset\mathbb{R}^d$ is a bounded space.
The following theorem shows nonexistence of an LSH family for NNS over $d_w^{l_1}$ over $\mathcal{X}=\mathcal{Y}=\mathcal{S}$.
\begin{theorem}
	\label{no lsh/alsh over M}
	For any $d>0$, $R_1<R_2$ and $P_1>P_2$, there is no $(R_1,R_2,P_1,P_2)$-sensitive LSH family for NNS over $d_w^{l_1}$ over $\mathcal{X}=\mathcal{Y}=\mathcal{S}$.
\end{theorem}
\begin{proof}
	Assume by contradiction that for some $d>0$, $R_1<R_2$ and $P_1>P_2$ there exists an $(R_1,R_2,P_1,P_2)$-sensitive LSH family $\mathcal{H}_{(h)}$ for NNS over $d_w^{l_1}$ over $\mathcal{S}$.
	Let $o,q\in\mathcal{S}$ where $o\neq q$\footnote{Ignore the trivial case that $\mathcal{S}$ contains only a single point.}. 
	As $w\in\mathbb{R}^d$, we can always set $w$ to a value such that $d_w^{l_1}(o,q)=R_1$ and thus $Pr[h(o)=h(q)]\geq P_1$. Moreover, we can always set $w$ to another value such that $d_w^{l_1}(o,q)=R_2$ and thus $Pr[h(o)=h(q)]\leq P_2$. However, since data points should be hashed before queries arrive, $\mathcal{H}_{(h)}$ can not involve $w$. So $Pr[h(o)=h(q)]$ is not affected by $w$, which leads to a contradiction.
\end{proof}
Due to the negative results in Theorems \ref{no lsh/alsh over Rd} and \ref{no lsh/alsh over M}, we seek to propose ALSH families for NNS over $d_w^{l_1}$ over bounded spaces in Section \ref{our solution}. 
Notice that if an ALSH family is suitable for NNS over $d_w^{l_1}$ over $\mathcal{X}=\mathcal{Y}=[M_{l},M_{u}]^d$ ($M_{l}<M_{u}$), it must also be suitable for NNS over $d_w^{l_1}$ over $\mathcal{X}=\mathcal{Y}=\mathcal{S}$ for any $\mathcal{S}\subseteq[M_{l},M_{u}]^d$. 
Thus, it is sufficient to deal with the case of $\mathcal{X}=\mathcal{Y}=[M_{l},M_{u}]^d$. Further, suppose $\mathcal{X}=\mathcal{Y}=[0,M_u-M_l]^d$. Otherwise, it can be satisfied by shifting $o,q\in[M_{l},M_{u}]^d$ without changing the results of NNS over $d_w^{l_1}$.
\section{Our Solutions}
\label{our solution}
Let $M=\lfloor (M_u-M_l)t\rfloor$ ($t>0$).
The following Observation \ref{observation 1} indicates that if we find an ALSH family for NNS over $d_w^{l_1}$ over $\mathcal{X}=\mathcal{Y}=\{0,1,\ldots,M\}^d$, a similar ALSH family can be immediately obtained for NNS over $d_w^{l_1}$ over $\mathcal{X}=\mathcal{Y}=[0,M_u-M_l]^d$.
Thus, we only need to consider NNS over $d_w^{l_1}$ over $\mathcal{X}=\mathcal{Y}=\{0,1,\ldots,M\}^d$ in the rest of the paper. Note that in our solutions $M$ can be an arbitrary positive integer.
\newtheorem{observation}{Observation}
\begin{observation}
	\label{observation 1}
	Define a vector function: $u_t(x)=\lfloor xt\rfloor=(\lfloor x_1t\rfloor,\lfloor x_2t\rfloor,\ldots,\lfloor x_dt\rfloor)$, where $x=(x_1,x_2,\ldots,x_d)\in[0,M_u-M_l]^d$ and $t>0$.
	For any $d>0$, $R_1<R_2$ and $P_1>P_2$,
	if $\mathcal{H}_{(f,g)}$ is an $(R_1,R_2,P_1,P_2)$-sensitive ALSH family for NNS over $d_w^{l_1}$ over $\mathcal{X}=\mathcal{Y}=\{0,1,\ldots,\lfloor(M_u-M_l)t\rfloor\}^d$, then  
	$\mathcal{H}_{(f\circ u_t,g\circ u_t)}$ must be
	an $(R_1',R_2',P_1,P_2)$-sensitive ALSH family for NNS over $d_w^{l_1}$ over $\mathcal{X}=\mathcal{Y}=[0,M_u-M_l]^d$,
	where $R_1'=(R_1-\sum_{i=1}^{d}\left|w_i\right|)/t$ and
	$R_2'=(R_2+\sum_{i=1}^{d}\left|w_i\right|)/t$.
\end{observation}
\begin{proof}
	Let $o,q\in[0,M_u-M_l]^d$. Then $\lfloor ot\rfloor,\lfloor qt\rfloor\in\{0,1,\ldots,\lfloor(M_u-M_l)t\rfloor\}^d$.
	A simple calculation shows that $d_w^{l_1}(\lfloor ot\rfloor,\lfloor qt\rfloor)-\sum_{i=1}^{d}\left|w_i\right|\leq d_w^{l_1}(ot,qt)=td_w^{l_1}(o,q)\leq d_w^{l_1}(\lfloor ot\rfloor,\lfloor qt\rfloor)+\sum_{i=1}^{d}\left|w_i\right|$ holds.
	Thus, we have that $d_w^{l_1}(\lfloor ot\rfloor,\lfloor qt\rfloor)\leq R_1$ if $d_w^{l_1}(o,q)\leq R_1'$ and $d_w^{l_1}(\lfloor ot\rfloor,\lfloor qt\rfloor)\geq R_2$ if $d_w^{l_1}(o,q)\geq R_2'$. Further, since $f(\lfloor ot\rfloor)=f(u_t(o))=(f\circ u_t)(o)$ and $g(\lfloor qt\rfloor)=g(u_t(q))=(g\circ u_t)(q)$, we have that $Pr[(f\circ u_t)(o)=(g\circ u_t)(q)]\geq P_1$ if $d_w^{l_1}(o,q)\leq R_1'$ and $Pr[(f\circ u_t)(o)=(g\circ u_t)(q)]\leq P_2$ if $d_w^{l_1}(o,q)\geq R_2'$. As a result, $\mathcal{H}_{(f\circ u_t,g\circ u_t)}$ is an $(R_1',R_2',P_1,P_2)$-sensitive ALSH family for NNS over $d_w^{l_1}$ over $\mathcal{X}=\mathcal{Y}=[0,M_u-M_l]^d$ (note: $R_1'<R_2'$ always holds since $R_1<R_2$).
\end{proof}
\subsection{From NNS over $d_w^{l_1}$ to MIPS}
\label{convertion}
In the following, we take two steps to convert the problem of NNS over $d_w^{l_1}$ into the problem of MIPS.
As a result of these steps, a novel preprocessing transformation and query transformation are introduced for data points and query points, respectively.
The two transformations are essential to our solutions.

\textbf{Step 1: Convert NNS over $d_w^{l_1}$ into NNS over $d_w^H$}

The generalized weighted Hamming distance $d_w^H$ is defined on the Hamming space and computed in the same way as the generalized weighted Manhattan distance $d_w^{l_1}$. That is, $d_w^H(o,q)=\sum_{i=1}^{d}w_i\left|o_i-q_i\right|$
for any $w=(w_1,w_2,\ldots,w_d)\in\mathbb{R}^d$, $o=(o_1,o_2,\ldots,o_d)\in\{0,1\}^d$ and $q=(q_1,q_2,\ldots,q_d)\in\{0,1\}^d$.

Inspired by \cite{DBLP:conf/vldb/GionisIM99}, we complete this step by applying unary coding.
Specifically, each point $x=(x_1,x_2,\ldots,x_d)\in\{0,1,\ldots,M\}^d$ is mapped into a binary vector
$v(x)=(\text{Unary}(x_1);\text{Unary}(x_2);\ldots;\text{Unary}(x_d))\in\{0,1\}^{Md}$,
where (;) is the concatenation and each $\text{Unary}(x_i)=(x_{i1},x_{i2},\ldots,x_{iM})$ is the unary representation of $x_i$, i.e., a sequence of $x_i$ 1's followed by $(M-x_i)$ 0's. Then $\left|o_i-q_i\right|=\sum_{j=1}^{M}\left|o_{ij}-q_{ij}\right|$ for $o=(o_1,o_2,\ldots,o_d)\in\{0,1,\ldots,M\}^d$, $q=(q_1,q_2,\ldots,q_d)\in\{0,1,\ldots,M\}^d$ and $1\leq i\leq d$.
Moreover, the weight vector $w=(w_1,w_2,\ldots,w_d)$ is mapped into $I(w)=(I(w_1);I(w_2);\ldots;I(w_d))$, where each $I(w_i)$ is a sequence of $M$ $w_i$'s.
As a result, we have
\begin{equation}
\label{gwmdtogwhd}
\begin{split}
d_w^{l_1}(o,q)=&\sum_{i=1}^{d}w_i\left|o_i-q_i\right|\\
=&\sum_{i=1}^{d}w_i(\sum_{j=1}^{M}\left|o_{ij}-q_{ij}\right|)\\
=&\sum_{i=1}^{d}\sum_{j=1}^{M}w_i\left|o_{ij}-q_{ij}\right|\\
=&d_w^H(v(o),v(q)),
\end{split}
\end{equation}
where $w=(w_1,w_2,\ldots,w_d)\in\mathbb{R}^d$, $o=(o_1,o_2,\ldots,o_d)\in\{0,1,\ldots,M\}^d$ and $q=(q_1,q_2,\ldots,q_d)\in\{0,1,\ldots,M\}^d$.
Equation \ref{gwmdtogwhd} indicates that through the above mappings NNS over $d_w^{l_1}$ over $\mathcal{X}=\mathcal{Y}=\{0,1,\ldots,M\}^d$ is converted into NNS over $d_w^H$ over $\mathcal{X}=\mathcal{Y}=\{v(x)\mid x\in\{0,1,\ldots,M\}^d\}\subset\{0,1\}^{Md}$.

\textbf{Step 2: Convert NNS over $d_w^H$ into MIPS} 

This step is based on the following observation.
\begin{observation}
	\label{observation 2}
	For any $o_{ij},q_{ij}\in\{0,1\}$ and $w_i\in\mathbb{R}$, the equation $w_i\left|o_{ij}-q_{ij}\right|=w_i-\left(\cos(\frac{\pi}{2}o_{ij}),\sin(\frac{\pi}{2}o_{ij})\right)^T\left(w_i\cos(\frac{\pi}{2}q_{ij}),w_i\sin(\frac{\pi}{2}q_{ij})\right)$ always holds.
\end{observation}
\begin{proof}
	We only need to check two cases. Case 1: If $o_{ij}=q_{ij}$, then $w_i\left|o_{ij}-q_{ij}\right|=0=w_i-\left(\cos(\frac{\pi}{2}o_{ij}),\sin(\frac{\pi}{2}o_{ij})\right)^T\left(w_i\cos(\frac{\pi}{2}q_{ij}),w_i\sin(\frac{\pi}{2}q_{ij})\right)$. Case 2: If $o_{ij}\neq q_{ij}$, then $w_i\left|o_{ij}-q_{ij}\right|=w_i=w_i-\left(\cos(\frac{\pi}{2}o_{ij}),\sin(\frac{\pi}{2}o_{ij})\right)^T\left(w_i\cos(\frac{\pi}{2}q_{ij}),w_i\sin(\frac{\pi}{2}q_{ij})\right)$.
\end{proof}

For any $x=(x_1,x_2,\ldots,x_d)\in\{0,1,\ldots,M\}^d$, we define $\widetilde{\cos}\left(\frac{\pi}{2}v(x)\right)$ and $\widetilde{\sin}\left(\frac{\pi}{2}v(x)\right)$ as follows:
\begin{gather}
\label{key1}
\widetilde{\cos}\left(\frac{\pi}{2}v(x)\right)=\left(\widehat{\cos}\left(\frac{\pi}{2}\text{Unary}(x_{1})\right);\widehat{\cos}\left(\frac{\pi}{2}\text{Unary}(x_{2})\right);\ldots;\widehat{\cos}\left(\frac{\pi}{2}\text{Unary}(x_d)\right)\right)\\
\widetilde{\sin}\left(\frac{\pi}{2}v(x)\right)=\left(\widehat{\sin}\left(\frac{\pi}{2}\text{Unary}(x_{1})\right);\widehat{\sin}\left(\frac{\pi}{2}\text{Unary}(x_{2})\right);\ldots;\widehat{\sin}\left(\frac{\pi}{2}\text{Unary}(x_d)\right)\right),
\end{gather}
where
\begin{gather}
\widehat{\cos}\left(\frac{\pi}{2}\text{Unary}(x_i)\right)=\left(\cos\left(\frac{\pi}{2}x_{i1}\right),\cos\left(\frac{\pi}{2}x_{i2}\right),\ldots,\cos\left(\frac{\pi}{2}x_{iM}\right)\right)\\
\label{key2}\widehat{\sin}\left(\frac{\pi}{2}\text{Unary}(x_i)\right)=\left(\sin\left(\frac{\pi}{2}x_{i1}\right),\sin\left(\frac{\pi}{2}x_{i2}\right),\ldots,\sin\left(\frac{\pi}{2}x_{iM}\right)\right).
\end{gather}

According to Observation \ref{observation 2}, we have
\begin{equation}
\label{gwhdtoip}
\begin{split}
d_w^H(v(o),v(q))=&\sum_{i=1}^{d}\sum_{j=1}^{M}w_i\left|o_{ij}-q_{ij}\right|\\
=&\sum_{i=1}^{d}\sum_{j=1}^{M}\left(w_i-\left(\cos(\frac{\pi}{2}o_{ij}),\sin(\frac{\pi}{2}o_{ij})\right)^T\left(w_i\cos(\frac{\pi}{2}q_{ij}),w_i\sin(\frac{\pi}{2}q_{ij})\right)\right)\\
=&M\sum_{i=1}^{d}w_i-d^{IP}\left(T^1\left(v(o)\right),T^2\left(v(q)\right)\right),\\
\end{split}
\end{equation}
where $w=(w_1,w_2,\ldots,w_d)\in\mathbb{R}^d$, $o=(o_1,o_2,\ldots,o_d)\in\{0,1,\ldots,M\}^d$, $q=(q_1,q_2,\ldots,q_d)\in\{0,1,\ldots,M\}^d$, $d^{IP}(\cdot,\cdot)$ is the inner product of two vectors, and $T^1(v(o))$ and $T^2(v(q))$ are respectively as follows:
\begin{gather}
\label{transformation1}
T^1(v(o))=(\widetilde{\cos}(\frac{\pi}{2}v(o));\widetilde{\sin}(\frac{\pi}{2}v(o)))\\\label{transformation2}
T^2(v(q))=(I(w)\odot\widetilde{\cos}(\frac{\pi}{2}v(q));I(w)\odot \widetilde{\sin}(\frac{\pi}{2}v(q))).
\end{gather}
From Equation \ref{gwhdtoip}, we can see that NNS over $d_w^H$ over $\mathcal{X}=\mathcal{Y}=\{v(x)\mid x\in\{0,1,\ldots,M\}^d\}\subset\{0,1\}^{Md}$ can be converted into MIPS over $\mathcal{X}=\{T^1(v(x))\mid x\in\{0,1,\ldots,M\}^d\}\subset\{0,1\}^{2Md}$ and $\mathcal{Y}=\{T^2(v(x))\mid x\in\{0,1,\ldots,M\}^d\}\subset\mathbb{R}^{2Md}$.

To sum up, after Steps 1 and 2, we convert NNS over $d_w^{l_1}$ into MIPS by using the composite functions $T^1(v(\cdot))$ and $T^2(v(\cdot))$ that respectively map data points and query points from $\{0,1,\ldots,M\}^d$ into two higher-dimensional spaces.
Let $P(o)=T^1(v(o))$ and $Q_w(q)=T^2(v(q))$ for $o,q\in\{0,1,\ldots,M\}^d$.
The vector functions $P(\cdot)$ and $Q_w(\cdot)$ are respectively the preprocessing and query transformations for the ALSH families introduced later. 
\subsection{ALSH Schemes for NNS over $d_w^{l_1}$}
Next, we formally present two ALSH schemes for NNS over $d_w^{l_1}$: the first one is called ($d_w^{l_1},l_2$)-ALSH and the second one is called ($d_w^{l_1},\theta$)-ALSH. ($d_w^{l_1},l_2$)-ALSH solves the problem of NNS over $d_w^{l_1}$ by reducing it to the problem of NNS over the $l_2$ distance, while ($d_w^{l_1},\theta$)-ALSH solves the problem of NNS over $d_w^{l_1}$ by reducing it to the problem of NNS over the Angular distance.

\subsubsection{($d_w^{l_1},l_2$)-ALSH}
Based on the transformations $P(\cdot)$ and $Q_w(\cdot)$ and the LSH family $\mathcal{H}_{(h^{l_2})}$ introduced in Section \ref{lsh}, ($d_w^{l_1},l_2$)-ALSH 
uses the ALSH family $\mathcal{H}_{(f^{(d_w^{l_1},l_2)},g^{(d_w^{l_1},l_2)})}=\{f^{(d_w^{l_1},l_2)}:\{0,1,\ldots,M\}^d\rightarrow\mathbb{Z}\}\bigcup\{g^{(d_w^{l_1},l_2)}:\{0,1,\ldots,M\}^d\rightarrow\mathbb{Z}\}$, where $f^{(d_w^{l_1},l_2)}(x)=h^{l_2}(P(x))$ and $g^{(d_w^{l_1},l_2)}(x)=h^{l_2}(Q_w(x))$ for $x\in\{0,1,\ldots,M\}^d$.
Combining Equations \ref{gwmdtogwhd} and \ref{gwhdtoip} we obtain
\begin{equation}
\label{gwhdtoiprelation}
d_w^{l_1}(o,q)=M\sum_{i=1}^{d}w_i-d^{IP}(P(o),Q_w(q)).
\end{equation}
It is easy to know
\begin{gather}
\label{Tnorm1}
\Vert P(o)\Vert_2^2=Md\\\label{Tnorm2}
\Vert Q_w(q)\Vert_2^2=M\sum_{i=1}^{d}w_i^2.
\end{gather}
Thus, we have
\begin{equation}
\label{dwl1tol2}
\begin{split}
d^{l_2}(P(o),Q_w(q))&=\Vert P(o)-Q_w(q)\Vert_2\\&=\sqrt{\Vert P(o)\Vert_2^2+\Vert Q_w(q)\Vert_2^2-2d^{IP}(P(o),Q_w(q))}\\&=\sqrt{M\left(d+\sum_{i=1}^{d}w_i^2\right)-2\left(M\sum_{i=1}^{d}w_i-d_w^{l_1}(o,q)\right)}.
\end{split}
\end{equation}
Let $r=d_w^{l_1}(o,q)$.
According to Equations \ref{collisonl2} and \ref{dwl1tol2}, the collision probability function with respect to $\mathcal{H}_{(f^{(d_w^{l_1},l_2)},g^{(d_w^{l_1},l_2)})}$ is
\begin{equation}
\label{l2_P}
\begin{split}
P^{(d_w^{l_1},l_2)}(r)&=Pr[f^{(d_w^{l_1},l_2)}(o)=g^{(d_w^{l_1},l_2)}(q)]\\&=Pr[h^{l_2}(P(o))=h^{l_2}(Q_w(q))]\\&=P^{l_2}(\Vert P(o)-Q_w(q)\Vert_2)\\&=P^{l_2}\left(\sqrt{M\left(d+\sum_{i=1}^{d}w_i^2\right)-2\left(M\sum_{i=1}^{d}w_i-r\right)}\right).
\end{split}
\end{equation}
Since $P^{l_2}(\cdot)$ is a decreasing function, $P^{(d_w^{l_1},l_2)}(R_1)>P^{(d_w^{l_1},l_2)}(R_2)$ holds for any $R_1<R_2$. Therefore, we obtain the following Lemma \ref{dl1tol2alsh}.
\newtheorem{lemma}{Lemma}
\begin{lemma}
	\label{dl1tol2alsh}
	$\mathcal{H}_{(f^{(d_w^{l_1},l_2)},g^{(d_w^{l_1},l_2)})}$ is $(R_1,R_2,P^{(d_w^{l_1},l_2)}(R_1),P^{(d_w^{l_1},l_2)}(R_2))$-sensitive for any $R_1<R_2$.
\end{lemma}
According to Theorem \ref{sublinear} and Lemma \ref{dl1tol2alsh}, we have the following Theorem \ref{l2/complexity}.
\begin{theorem}
	\label{l2/complexity}
	($d_w^{l_1},l_2$)-ALSH can solve the problem of $(R_1,R_2)$-NNS over $d_w^{l_1}$ with $O(n^{\rho^{(d_w^{l_1},l_2)}}d\log n)$ query time and $O(n^{1+\rho^{(d_w^{l_1},l_2)}})$ space, where $\rho^{(d_w^{l_1},l_2)}=\frac{\log\left(P^{(d_w^{l_1},l_2)}(R_1)\right)}{\log \left(P^{(d_w^{l_1},l_2)}(R_2)\right)}<1$.
\end{theorem}
\subsubsection{($d_w^{l_1},\theta$)-ALSH}
Now we introduce the scheme of ($d_w^{l_1},\theta$)-ALSH. Based on the transformations $P(\cdot)$ and $Q_w(\cdot)$ and the LSH family $\mathcal{H}_{(h^{\theta})}$ introduced in Section \ref{lsh}, ($d_w^{l_1},\theta$)-ALSH uses the ALSH family $\mathcal{H}_{(f^{(d_w^{l_1},\theta)},g^{(d_w^{l_1},\theta)})}=\{f^{(d_w^{l_1},\theta)}:\{0,1,\ldots,M\}^d\rightarrow\{0,1\}\}\bigcup\{g^{(d_w^{l_1},\theta)}:\{0,1,\ldots,M\}^d\rightarrow\{0,1\}\}$, where $f^{(d_w^{l_1},\theta)}(x)=h^{\theta}(P(x))$ and $g^{(d_w^{l_1},\theta)}(x)=h^{\theta}(Q_w(x))$ for $x\in\{0,1,\ldots,M\}^d$.
According to Equations \ref{gwhdtoiprelation}, \ref{Tnorm1} and \ref{Tnorm2}, the relationship between $d_w^{l_1}(o,q)$ and $d^{\theta}(P(o),Q_w(q))$is as follows:
\begin{equation}
\label{dwl1totheta}
\begin{split}
d^{\theta}(P(o),Q_w(q))&=\arccos\left(\frac{P(o)^TQ_w(q)}{\Vert P(o)\Vert_2\Vert Q_w(q)\Vert_2}\right)\\&=\arccos\left(\frac{d^{IP}(P(o),Q_w(q))}{\Vert P(o)\Vert_2\Vert Q_w(q)\Vert_2}\right)\\&=\arccos\left(\frac{M\sum_{i=1}^{d}w_i-d_w^{l_1}(o,q)}{M\sqrt{d\sum_{i=1}^{d}w_i^2}}\right).
\end{split}
\end{equation}

Let $r=d_w^{l_1}(o,q)$.
From Equations \ref{collisontheta} and \ref{dwl1totheta}, it can be seen that the collision probability function with respect to $\mathcal{H}_{(f^{(d_w^{l_1},\theta)},g^{(d_w^{l_1},\theta)})}$ is
\begin{equation}
\label{theta_P}
\begin{split}
P^{(d_w^{l_1},\theta)}(r)&=Pr[f^{(d_w^{l_1},\theta)}(o)=g^{(d_w^{l_1},\theta)}(q)]\\&=Pr[h^{\theta}(P(o))=h^{\theta}(Q_w(q))]\\&=1-\frac{1}{\pi}\arccos\left(\frac{P(o)^TQ_w(q)}{\Vert P(o)\Vert_2\Vert Q_w(q)\Vert_2}\right)\\&=1-\frac{1}{\pi}\arccos\left(\frac{M\sum_{i=1}^{d}w_i-r}{M\sqrt{d\sum_{i=1}^{d}w_i^2}}\right).
\end{split}
\end{equation}
It is easy to know that $P^{(d_w^{l_1},\theta)}(\cdot)$ is a decreasing function. Thus, $P^{(d_w^{l_1},\theta)}(R_1)>P^{(d_w^{l_1},\theta)}(R_2)$ holds for any $R_1<R_2$. Then we obtain the following Lemma \ref{dl1tothetaalsh}.
\begin{lemma}
	\label{dl1tothetaalsh}
	$\mathcal{H}_{(f^{(d_w^{l_1},\theta)},g^{(d_w^{l_1},\theta)})}$ is $(R_1,R_2,P^{(d_w^{l_1},\theta)}(R_1),P^{(d_w^{l_1},\theta)}(R_2))$-sensitive for any $R_1<R_2$.
\end{lemma}
Combining Theorem \ref{sublinear} and Lemma \ref{dl1tothetaalsh} we have the following Theorem \ref{theta/complexity}.
\begin{theorem}
	\label{theta/complexity}
	($d_w^{l_1},\theta$)-ALSH can solve the problem of $(R_1,R_2)$-NNS over $d_w^{l_1}$ with $O(n^{\rho^{(d_w^{l_1},\theta)}}d\log n)$ query time and $O(n^{1+\rho^{(d_w^{l_1},\theta)}})$ space, where $\rho^{(d_w^{l_1},\theta)}=\frac{\log\left(P^{(d_w^{l_1},\theta)}(R_1)\right)}{\log \left(P^{(d_w^{l_1},\theta)}(R_2)\right)}<1$.
\end{theorem}
\subsubsection{Implementation Skills of ($d_w^{l_1},l_2$)-ALSH and ($d_w^{l_1},\theta$)-ALSH}
The scheme of ($d_w^{l_1},l_2$)-ALSH (or ($d_w^{l_1},\theta$)-ALSH) needs to compute the hash values $h^{l_2}(P(o))$ and $h^{l_2}(Q_w(q))$ (or $h^{\theta}(P(o))$ and $h^{\theta}(Q_w(q))$).
We can easily know that the running time of computing $h^{l_2}(P(o))$ (or $h^{\theta}(P(o))$) is dominated by the time cost of obtaining $a^TP(o)$ and the running time of computing $h^{l_2}(Q_w(q))$ (or $h^{\theta}(Q_w(q))$) is dominated by the time cost of obtaining $a^TQ_w(q)$, where $a$ is a $2Md$-dimensional vector where each entry is chosen independently from the standard normal distribution. 
The naive approach to obtain $a^TP(o)$ or $a^TQ_w(q)$ is to compute the inner product of the two corresponding vectors. 
However, it will require $2Md$ multiplications and $2Md-1$ additions, which is expensive when $M$ is large.

Next, we show how to obtain $a^TP(o)$ with only $2d-1$ additions and obtain $a^TQ_w(q)$ with only $2d-1$ additions and $d$ multiplications.
Suppose $a=(a_1;a_2;\ldots;a_d;a_{d+1};a_{d+2};\ldots;a_{2d})$, where $a_i=(a_{i1},a_{i2},\ldots,a_{iM})\in\mathbb{R}^M$.
According to Equations \ref{key1}-\ref{key2}, \ref{transformation1} and \ref{transformation2}, we have
$a^TP(o)=\sum_{i=1}^{d}a_i^T\widehat{\cos}(\frac{\pi}{2}\text{Unary}(o_i))+\sum_{i=1}^{d}a_{d+i}^T\widehat{\sin}(\frac{\pi}{2}\text{Unary}(o_i))$ and $a^TQ_w(q)=\sum_{i=1}^{d}a_i^T\widehat{\cos}(\frac{\pi}{2}\text{Unary}(q_i))w_i+\sum_{i=1}^{d}a_{d+i}^T\widehat{\sin}(\frac{\pi}{2}\text{Unary}(q_i))w_i$.
Since $\widehat{\cos}(\frac{\pi}{2}\text{Unary}(o_i))$ is a sequence of $o_i$ 0's followed by $(M-o_i)$ 1's and
$\widehat{\sin}(\frac{\pi}{2}\text{Unary}(o_i))$ is a sequence of $o_i$ 1's followed by $(M-o_i)$ 0's, it is easy to know that $a_i^T\widehat{\cos}(\frac{\pi}{2}\text{Unary}(o_i))$ is the sum of the last $M-o_i$ elements of $a_i$ and
$a_{d+i}^T\widehat{\sin}(\frac{\pi}{2}\text{Unary}(o_i))$ is the sum of the first $o_i$ elements of $a_{d+i}$.
Thus, we preprocess the vector $a$ to obtain $a'=(a_1';a_2';\ldots;a_d';a_{d+1}';a_{d+2}';\ldots;a_{2d}')$, where $a_i'=(a_{i1}',a_{i2}',\ldots,a_{iM}',a_{i(M+1)}')$ and
\begin{equation}
a_{ij}'=\begin{cases}
\sum_{k=j}^{M}a_{ik}&\text{if}\ 1\leq i\leq d\text{ and }1\leq j\leq M\\
0&\text{if}\ 1\leq i\leq d\text{ and }j=M+1\\
0&\text{if}\ d+1\leq i\leq 2d\text{ and }j=1\\
\sum_{k=1}^{j-1}a_{ik}&\text{if}\ d+1\leq i\leq 2d\text{ and }2\leq j\leq M+1.\\
\end{cases}
\end{equation}
Then we have $a^TP(o)=\sum_{i=1}^{d}a_{i(o_i+1)}'+\sum_{i=1}^{d}a_{(d+i)(o_i+1)}'$. It can be seen that $a^TP(o)$ can be obtained with $2d-1$ additions by using $a'$.
Similarly, we have $a^TQ_w(q)=\sum_{i=1}^{d}w_ia_{i(q_i+1)}'+\sum_{i=1}^{d}w_{i}a_{(d+i)(q_i+1)}'=\sum_{i=1}^{d}w_i(a_{i(q_i+1)}'+a_{(d+i)(q_i+1)}')$. Therefore, $a^TQ_w(q)$ can be obtained with $2d-1$ additions and $d$ multiplications by using $a'$.

\section{Conclusion}
This paper studies the fundamental problem of Nearest Neighbor Search (NNS) over the generalized weighted Manhattan distance ($d_w^{l_1}$). As far as we know, there is no prior work that solves the problem in sublinear time. In this paper, we first prove that there is no LSH or ALSH family for $d_w^{l_1}$ over the entire space $\mathbb{R}^d$. Then, we prove that there is still no LSH family suitable for $d_w^{l_1}$ over a bounded space. After that, we propose two ALSH families for $d_w^{l_1}$ over a bounded space. Based on these ALSH families, two ALSH schemes ($d_w^{l_1},l_2$)-ALSH and ($d_w^{l_1},\theta$)-ALSH are proposed for solving NNS over $d_w^{l_1}$ in sublinear time.

\bibliography{neurips_2021}
\bibliographystyle{plain}
%
%
%
%



%
%

\end{document}